\documentclass[reqno,11pt]{article} 
\usepackage[margin=1in]{geometry}
\usepackage{amssymb}
\usepackage{amsmath}
\usepackage{amsthm}
\usepackage{xspace}
\usepackage{tikz}
\usepackage{hyperref}

\newtheorem{theorem}{Theorem}
\newtheorem{lemma}[theorem]{Lemma}
\newtheorem{proposition}[theorem]{Proposition}
\newtheorem{corollary}[theorem]{Corollary}

\theoremstyle{remark}

\newcommand{\RtoL}{\text{RevtoL}\xspace} 

\newcommand{\ALG}{\textsc{ALG}\xspace}
\newcommand{\OPT}{\textsc{OPT}\xspace}

\begin{document}

\title{Online Interval Selection on a Simple Chain}

\author{
Yaqiao Li\footnote{Shenzhen University of Advanced Technology, Shenzhen, China, liyaqiao@suat-sz.edu.cn},
Ali Mohammad Lavasani\footnote{Concordia University, Montr{\'e}al, Canada, ali.mohammadlavasani@concordia.ca},
Denis Pankratov\footnote{Concordia University, Montr{\'e}al, Canada, denis.pankratov@concordia.ca}}

\maketitle

\begin{abstract}
A set of intervals $I = \{ I_1, I_2, \dots, I_n \}$ forms a simple chain if, for every $2\leq i \leq n-1$, interval $I_i$ overlaps only with $I_{i-1}$ and $I_{i+1}$.  
We show that a  deterministic memoryless one-directional revoking algorithm achieves a competitive ratio of $2(1 - 1/\sqrt{e}) \approx 0.786$ on the simple chain in the random order model, hence performs worse than the basic greedy algorithm without revoking that has a competitive ratio of $(1 - 1/e^2) \approx 0.864$, but better than any deterministic revoking algorithm in the adversarial model that has a competitive ratio of at most $0.75$. The proof of the latter also leads to a lower bound of $n/4$ for the advice complexity. 
\end{abstract}

\section{Introduction}

In the Online Interval Selection problem, a set of intervals $I = \{ I_1, I_2, \dots, I_n \}$ arrives one by one, an algorithm must decide whether to accept or reject each interval as it arrives, ensuring that the selected intervals do not overlap. 
In the unweighted version of this problem, discussed in this work, the objective is to maximize the number of selected intervals. In the revocable acceptances setting, the algorithm is allowed to revoke previously selected intervals to accept new ones, while maintaining a set of non-overlapping intervals at all times. 
Interval selection has many applications, such as resource allocation, network routing,  etc, see surveys by
Kolen et al. \cite{kolen2007interval} and Kovalyov et al \cite{kovalyov2007fixed}.

The performance of an online algorithm is measured by its \emph{competitive ratio}. For an online maximization problem, Let $\ALG(I)$ denote the performance of algorithm $\ALG$ on input $I$, and let $\OPT$ be an optimal offline algorithm. We say  $\ALG$ has a competitive ratio of at least $\rho$ if there exists a constant $b$ such that $\ALG(I) \geq \rho \cdot \OPT(I) - b$ for all inputs $I$. The competitive ratio of $\ALG$ is defined as the largest $\rho$ for which this inequality holds.

Recently, Borodin and Karavasilis~\cite{borodin2023any} introduced the \emph{(any order) adversarial}   model,  where  intervals are allowed to arrive in any possible order.  They also discussed the \emph{random order model}, where the set of intervals is generated by an adversary and then given to the algorithm in a uniform random order. We study both models  and focus  on the input that is a \textit{simple chain}: $n$ unit-length intervals represent a path graph of $n$ vertices, see Figure \ref{fig:chain}.

\begin{figure}[ht!]        
    \centering
    \includegraphics[scale=.9]{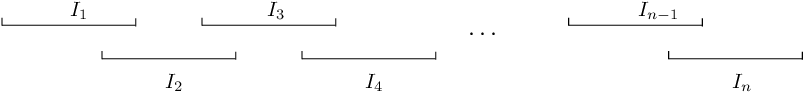}
    \caption{A simple chain of  intervals of size $n$.}
\label{fig:chain}
\end{figure}

Our work is inspired by the following comment in  \cite{borodin2023any}, ``for single-lengthed instances, a one-directional algorithm  is the only deterministic memoryless algorithm that can possibly benefit from random arrivals.'' For such instance Borodin and Karavasilis~\cite{borodin2023any} gave a simple revoking algorithm that achieves a competitive ratio of $1/2$ in the adversarial model, and proved that every deterministic memoryless revoking algorithm that is not \emph{one-directional} has a competitive ratio of at most $1/2$ in the random order model.  Here, one-directional means either  only revoke intervals on the left side, or only on the right side.

Our results partially confirm the comment on the simple chain. We determine a one-directional revoking algorithm Revoke-to-the-Left achieves a competitive ratio of $2(1 - 1/\sqrt{e}) \approx 0.786$ in the random order model, and prove that in the adversarial model any deterministic revoking algorithm has an upper bound of competitive ratio of at most $3/4=0.75 < 0.786$. Though, we point out that Revoke-to-the-Left performs worse than  the basic greedy algorithm without revoking that has a competitive ratio of $(1 - 1/e^2) \approx 0.864$. Still,  the Revoke-to-the-Left algorithm remains more promising in more general interval graphs. The construction also leads to a lower bound of  $n/4$ for the advice complexity,  marking the first lower bound for this problem.  In the advice setting, an all-powerful trustworthy oracle, which has full knowledge of the input in advance, populates an advice tape with bits according to a pre-agreed protocol. The online algorithm has the option to read the advice bits to achieve the best possible performance against an adversary. The worst-case number of advice bits read by the online algorithm for inputs of length $n$ is the \emph{advice complexity} of the algorithm, see   a survey by Boyar et al \cite{boyar2016online_advice_survey}. Recently, Boyar et al \cite{boyar2023onlineinterval_predictions} and Karavasilis \cite{karavasilis2025interval_prediction} demonstrated the power of prediction, i.e., imperfect advice, on interval selection problems.

It would be interesting to generalize our results to $k$-chain, where each interval intersects with $k$ intervals on either side. This helps for a deeper understanding of the random order versus adversarial models, and it seems challenging.

\section{Revoke-to-the-Left in the random order model}

We first consider the basic greedy algorithm without revoking that selects an interval whenever it does not overlap other selected intervals.

\begin{theorem}
	The basic greedy algorithm has a competitive ratio of $(1-\frac{1}{e^2}) \approx 0.864$ in the random order model.
\end{theorem}

\begin{proof}
	Consider the  Unfriendly Seating Arrangement problem: There are $n$ seats in a row at a luncheonette and people sit down one at a time at random. They are unfriendly and so never sit next to one another (no moving over). Freedman and Shepp \cite{freedman1962unfriendly} showed  the expected number of persons to sit down  is  $n(1-\frac{1}{e^2})/2 + O(1)$. It is easy to see that there is a one-to-one mapping from the seats taken   to the intervals selected by the basic greedy algorithm. Since $\OPT$ is always $\lceil n/2 \rceil$, the theorem follows.
\end{proof}

\begin{figure}[ht!]        
    \centering
    \includegraphics[scale=.9]{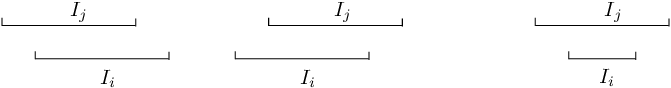}
    \caption{In all cases, $I_j$ is already selected when $I_i$ arrives. In the first, $I_i$ is rejected. In second and third, $I_j$ is revoked and $I_i$ is accepted. In our proof, the third case would not occur as our input is a simple chain.}  \label{fig:revtol}
\end{figure}

Now we consider the ``Revoke-to-the-Left'' (\RtoL) algorithm. It can be defined for general inputs, not only chains. When an interval $I_i = [s_i, f_i]$ arrives, if it overlaps with a selected interval $I_j = [s_j, f_j]$ such that $s_i < f_j < f_i$ (i.e., $I_j$ is on the left side of $I_i$), $I_i$ is rejected. Otherwise, $I_i$ is accepted, revoking any previously selected interval that overlaps with $I_i$ (see Figure \ref{fig:revtol}).

\begin{theorem}
\label{thm:rtol}
	The  algorithm \RtoL has a competitive ratio of $2(1 - \frac{1}{\sqrt{e}}) \approx 0.786$ in the random order model.
\end{theorem}

Let $I$ be a simple chain of $n$ intervals, numbered successively by interval $1$, \ldots, interval $n$.  Let $S_n$ be the set of all permutations on $[n]$. Let $\sigma \in S_n$. Let $\RtoL(I,\sigma)$ be the set of intervals chosen by \RtoL on  $I$ when they arrive in the order of $\sigma$.  
Let $d_{n,i} = \#\{\sigma \in S_n : \text{ interval } i \text{ is chosen by \RtoL} \}$, $d_n := d_{n,n}$. We are interested in $D_n = \sum_{\sigma \in S_n} \RtoL(I,\sigma)$, i.e., the expected number of intervals selected by \RtoL is $D_n/n!$. By double counting, 
$D_n := \sum_{i=1}^n d_{n,i}$. 
Table \ref{tab:data} shows data for small $n$.

\begin{table}[ht!]  
\centering
\caption{Data for $d_n, D_n$.} \label{tab:data} \label{tab:data}
\begin{tabular}{|c|c|c|c|c|c|c|}
\hline
$n$   & 1 & 2 & 3  & 4  & 5   & 6    \\ 
\hline
$d_n$ & 1 & 0 & 4  & 6  & 56  & 260  \\ 
\hline
$D_n$ & 1 & 2 & 10 & 46 & 286 & 1976 \\ 
\hline
$d_{n,n} - d_{n,n-1} $ & 1 & $-2$ & 4 & $-10$ & 26 & $-76$ \\ 
\hline
\end{tabular}
\end{table}

\begin{lemma}   \label{lem:one_step_reduction}
    For $1 \le i \le n-1$,
	$d_{n,i} = n d_{n-1,i} = n (n-1) \cdots (i+1) d_i$.
\end{lemma}

\begin{proof}
    It suffices to show 
        $d_{n,i} = n d_{n-1,i}$.
    This is because for every $1 \le i \le n-1$, the presence of interval $n$ does not influence whether interval $i$  is chosen or not chosen. Indeed, interval $n$ has the least priority and \RtoL does not revoke or reject any other interval because of interval $n$. Since interval $n$ can be in $n$ different positions, the equality follows.
\end{proof}

The following is the key equation for proving Theorem \ref{thm:rtol}.

\begin{proposition}   \label{thm:formula_dn_difference}
    For $n\ge 4$,
        $d_{n,n} - d_{n,n-1} = 2d_{n-1,n-2} - d_{n-1,n-1} - d_{n-1,n-3}$.
\end{proposition}

\begin{proof}
    We give a formula for $d_n$. For 
        $1 \le i \le n$, 
    define
        \begin{equation}    \label{eq:def_pi}
	\begin{split}
            p_i &= \#\{\sigma \in S_n: \text{ interval } n \text{ is  chosen by \RtoL,}  \\
			&\text{ while interval } i \text{ is the last interval of the input} \}.
	\end{split}
        \end{equation}
    Then,
        $d_n = p_1 + \ldots + p_n$.
    Below we give formulas for computing $p_i$.

 	(i) $p_n = (n-1)! - d_{n-1}$.             This is because interval $n$ is chosen iff interval $n-1$ is \emph{not} chosen among intervals
                $1, \ldots, (n-1)$.

       (ii)  $p_{n-1} = (n-1) d_{n-2} = d_{n-1,n-2}$.       This is because in this case interval $n$ is chosen if and only if interval $n-2$ is chosen among intervals 
                $1, \ldots, (n-2)$, 
            and we have the multiplicative factor $n-1$ because interval $n$ can be arbitrarily placed in any of the $n-1$ positions.

        (iii)  for $2 \le i \le n-2$, 
            $p_i = (n-1)\cdots(n-i+1)d_{n-i} = d_{n-1,n-i}$.
            
            Consider the two blocks of intervals: intervals 
                $1, \ldots, (i-1)$, 
            and intervals 
                $(i+1), \ldots, n$. 
            Since interval $i$ is not present when these $n-1$ intervals appear, we know that these two blocks of intervals do not interfere with each other. Hence, interval $n$ is chosen if and only if it is chosen among intervals 
                $(i+1), \ldots, n$. 
            But this is equivalent to interval $n-i$ is chosen among intervals 
                $1, \ldots, (n-i)$. 
            Since the other block of intervals 
                $1, \ldots, (i-1)$ 
            can be arbitrarily placed in $n-1$ positions, we have the multiplicative factor 
                $(n-1)\cdots(n-i+1)$.

        (iv)  $p_1 = d_{n-1}$.      This is because in this case interval $n$ is chosen if and only if interval $n$ is chosen among intervals 
                $2, \ldots, n$, 
            which is equivalent to interval $n-1$ is chosen among intervals 
                $1, \ldots, (n-1)$.

Hence, 
        $d_n
        = (n-1)! + 2d_{n-1,n-2} + \sum_{j=2}^{n-3} d_{n-1,j}$.
    With this, we have 
\begin{align*}
        &d_{n,n} - d_{n,n-1}= d_{n,n} - n d_{n-1,n-1} = d_{n,n} - (n-1) d_{n-1,n-1} - d_{n-1,n-1} \\
        &= \left( (n-1)! + 2d_{n-1,n-2} + \sum_{j=2}^{n-3} d_{n-1,j} \right) \\
        &\phantom{==} - (n-1) \left( (n-2)! + 2d_{n-2,n-3} + \sum_{j=2}^{n-4} d_{n-2,j} \right) 
        - d_{n-1,n-1}   \\
        &= \left( (n-1)! + 2d_{n-1,n-2} + \sum_{j=2}^{n-3} d_{n-1,j} \right) \\
        &\phantom{==} - \left( (n-1)! + 2d_{n-1,n-3} + \sum_{j=2}^{n-4} d_{n-1,j} \right)
        - d_{n-1,n-1}   \\
        &= 2d_{n-1,n-2} - d_{n-1,n-1} - d_{n-1,n-3}.	\qedhere
\end{align*}
\end{proof}

\begin{corollary}  \label{thm:formula_dn}
    For $n\ge 3$, 
        $d_n = (n-1)(D_{n-2} + d_{n-2})$.
\end{corollary}

\begin{proof}
	We use induction.  The base case can be verified by Table \ref{tab:data}. Now assume the claim is true for $n-1$, i.e., $d_{n-1} = (n-2)(D_{n-3} + d_{n-3})$, we prove it for $n$.  By Lemma \ref{lem:one_step_reduction} and the induction hypothesis, one has
       	$D_{n-2} = (n-2)D_{n-3} + d_{n-2}  = 	d_{n-1} - (n-2)d_{n-3} + 	d_{n-2}$.
	It suffices to check this is equal to
		$d_n/(n-1) - d_{n-2}$.
	Indeed, by Lemma \ref{lem:one_step_reduction}, one has
	$d_n - (n-1) d_{n-2} = d_n - d_{n-1, n-2}$,
and on the other hand, 
	$(n-1)D_{n-2} = (n-1)(d_{n-1} - (n-2)d_{n-3} + 	d_{n-2}) = (n-1)d_{n-1} - d_{n-1, n-3} + d_{n-1, n-2}$. 
Their difference is 
	$d_n - 2d_{n-1, n-2} - (n-1)d_{n-1} +d_{n-1, n-3} = 0$
by Proposition \ref{thm:formula_dn_difference}, as desired.
\end{proof}

Consider an auxiliary integer sequence $a_n = a_{n-1} + (n-1)a_{n-2}$, with initial conditions $a_0 = a_1 = 1$. The first few numbers  are $1, 1, 2, 4, 10, 26, 76, 232, \ldots$. This sequence is known to represent the number of self-inverse permutations as well as other  combinatorial objects \cite{oeis2024involutions}. We use  the recursive definition of $a_n$.

\begin{corollary}   \label{cor:dn_an}
    $d_{n,n} - d_{n,n-1} = (-1)^{n-1} a_n$.   
\end{corollary}

\begin{proof}
    The base case can be verified from Table \ref{tab:data}. Then, it suffices to prove 
        $a_n = (-1)^{n-1} (d_{n,n} - d_{n,n-1})$
    satisfy the recurrence formula of $a_n$, i.e.,
        $a_n = a_{n-1} + (n-1) a_{n-2}$.
    Indeed, by Proposition \ref{thm:formula_dn_difference} and Lemma \ref{lem:one_step_reduction},
    \begin{align*}
        d_{n,n} - d_{n,n-1}
        &= 2d_{n-1,n-2} - d_{n-1,n-1} - d_{n-1,n-3} \\
        &= - (d_{n-1,n-1} - d_{n-1,n-2}) + d_{n-1,n-2} - d_{n-1,n-3} \\
        &= - (d_{n-1,n-1} - d_{n-1,n-2}) + (n-1) (d_{n-2,n-2} - d_{n-2,n-3}).   \qedhere
    \end{align*}
\end{proof}

\begin{corollary}   \label{cor:Dn_recursion}
    $D_n = (n+1) D_{n-1} + (-1)^{n-1} a_{n-1}$.
\end{corollary}

\begin{proof}
    By Corollary \ref{thm:formula_dn}, Lemma \ref{lem:one_step_reduction}, and Corollary \ref{cor:dn_an}, we have
    \begin{align*}
        D_n 
        &= \sum_{i=1}^n d_{n,i} 
        = n \left(\sum_{i=1}^{n-1} d_{n-1,i} \right) + d_{n,n} 
        = n D_{n-1} + (n-1)(D_{n-2} + d_{n-2}) \\
        &= n D_{n-1} + D_{n-1} - d_{n-1,n-1} + d_{n-1,n-2}  = (n+1) D_{n-1} + (-1)^{n-1} a_{n-1}.  \qedhere
    \end{align*}
\end{proof}

\begin{corollary}
\label{cor:Dn_ratio}
    $\frac{D_n}{(n+1)!} = \sum_{i=0}^{n-1} \frac{(-1)^i a_i}{(i+2)!}$.
\end{corollary}

\begin{proof}
    We prove this by induction. The base case $n=1$ is trivial. Assume it is true for $n=k$. The induction step for $n=k+1$ can be verified by expanding $D_{k+1}$ via Corollary \ref{cor:Dn_recursion} as follows,
    \begin{align*}
    \frac{D_{k+1}}{(k+2)!} &= \frac{(k+2)D_k + (-1)^k a_k}{(k+2)!} = \frac{D_{k}}{(k+1)!} + \frac{(-1)^k a_k}{(k+2)!} \\
    &= \sum_{i=0}^{k-1} \frac{(-1)^i a_i}{(i+2)!} +\frac{(-1)^k a_k}{(k+2)!} =
    \sum_{i=0}^k \frac{(-1)^i a_i}{(i+2)!}. \qedhere
     \end{align*}
\end{proof}

\begin{proposition} \label{lem:lim}
	$\lim_{n\to \infty} \frac{D_n}{(n+1)!}= 1 - e^{-\frac{1}{2}}$.
\end{proposition}

\begin{proof}
Let
    $S = \lim_{n\to \infty} \frac{D_n}{(n+1)!}$. 
By Corollary \ref{cor:Dn_ratio},
    $S=\sum_{n=0}^{\infty} \frac{(-1)^n a_n}{(n+2)!}$.
We will employ generating functions and integral representations. Let 
    $G(x) = \sum_{n=0}^{\infty} \frac{a_n x^n}{n!}$ 
denote the exponential generating function for the sequence $\{a_n\}$. We claim 
        $G''(x) = (1+x) G'(x) + G(x)$. 
    Indeed, a direct calculation gives
        $G'(x) = \sum_{n=0}^\infty \frac{a_{n+1}}{n!}x^n$
    and
        $G''(x) = \sum_{n=0}^\infty \frac{a_{n+2}}{n!}x^n$.
    Hence, by the formula 
        $a_n = a_{n-1} + (n-1)a_{n-2}$,
    we have
    \begin{align*}
        G''(x) - G'(x)
        &= \sum_{n=0}^\infty \frac{a_{n+2} - a_{n+1}}{n!}x^n    
        = \sum_{n=0}^\infty \frac{(n+1)a_n}{n!}x^n  \\
        &= \left( \sum_{n=0}^\infty \frac{a_n}{n!}x^{n+1}  \right)' 
        = ( x G(x) )'
        = G(x) + xG'(x),
    \end{align*}
    as claimed.
    Solving the differential equation 
with the boundary condition
        $G(0) = 1$
    and 
        $G'(0) = 1$,
    we get $G(x) = e^{\frac{x(x+2)}{2}}$.

Next, we express $S$ in terms of the generating function $G(x)$. Observe that
    $S = \sum_{n=0}^{\infty} \frac{(-1)^n a_n}{n!} \cdot \frac{1}{(n+1)(n+2)}$.
To relate this to the generating function $G(x)$, we use the integral representation
    $\frac{1}{(n+1)(n+2)} = \int_0^1 x^n (1 - x) \, dx$.
Substituting this into the expression for $S$, we obtain
    $S = \sum_{n=0}^{\infty} \frac{(-1)^n a_n}{n!} \int_0^1 x^n (1 - x) \, dx = \int_0^1 (1 - x) \sum_{n=0}^{\infty} \frac{a_n (-x)^n}{n!} \, dx$.
Recognizing the sum inside the integral as the generating function $G(x)$ evaluated at $-x$, we have
	$\sum_{n=0}^{\infty} \frac{a_n (-x)^n}{n!} = G(-x) = e^{-x + \frac{x^2}{2}}$.
Therefore,
$S = \int_0^1 (1 - x) e^{-x + \frac{x^2}{2}} \, dx = 1 - e^{-\frac{1}{2}}$.
\end{proof}

\begin{proof}[Proof of Theorem \ref{thm:rtol}]
    Recall that by definition $D_n/n!$ is the expected number of intervals chosen by \RtoL. 
Since \OPT is always $\lceil n/2 \rceil$, by  Proposition \ref{lem:lim}, 
    \[
        \lim_{n\to \infty} \frac{D_n}{n!\lceil n/2 \rceil} = \lim_{n\to \infty} \frac{2D_n}{(n+1)!} = 2(1-\frac{1}{\sqrt{e}}).    \qedhere
    \]    
\end{proof}

We remark that the actual competitive ratio sequence $D_n/(n!\lceil n/2 \rceil)$, although itself alternate between odd and even terms,  both the odd subsequence and the even subsequence decrease monotonically to the limit.

\section{The adversarial model and the advice complexity}

To complement Theorem \ref{thm:rtol}, we prove the following upper bound for the adversarial model. Incidentally, the construction also leads to a lower bound for advice complexity.

\begin{figure}[ht!]        
    \centering
    \includegraphics[scale=.9]{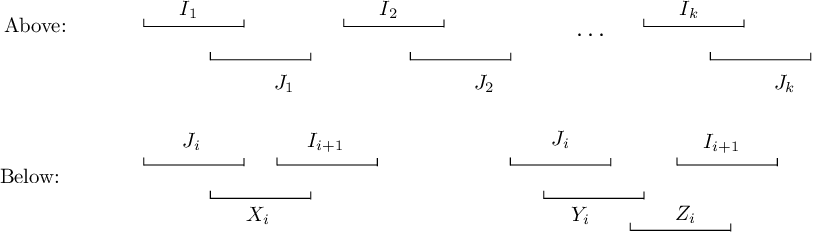}
    \caption{Above: the first $2k$ intervals. Below: Connect $J_i$ and $I_{i+1}$ by putting one interval or two intervals, depending on the \textit{gap decision} of \ALG.} \label{fig:pairs_gaps}   
\end{figure}

\begin{theorem}
    \label{thm:interval_rev_ub}
    No deterministic algorithm with revoking can achieve an asymptotic competitive ratio better than $3/4$ in the adversarial model. 
\end{theorem}

\begin{proof}
    Fix a deterministic revoking algorithm \ALG. The adversarial input begins with $k$ pairs of intervals $(I_1, J_1), (I_2, J_2), \ldots, (I_k, J_k)$ such that $I_i$ overlaps $J_i$ on the left, and different pairs are disjoint, see Figure \ref{fig:pairs_gaps}. Because of revoking, without loss of generality we can assume that \ALG chooses exactly one interval in each pair, hence, up to this point at most $k$ intervals are selected.

Consider the \ALG decision on $J_i$ and $I_{i+1}$ and we call it $i^\text{th}$ \textit{gap decision}. 
Let A denote accept and R reject,  a gap decision can be: RA, AR, RR, or AA. In the case of RA or AR, we introduce a new interval $X_i$ that intersects only these two intervals, see Figure \ref{fig:pairs_gaps}. In the case of RR or AA, we introduce two new intervals $Y_i$ and $Z_i$, where $Y_i$ intersects only $J_i$ and $Z_i$, and $Z_i$ intersects only $Y_i$ and $I_{i+1}$, see Figure \ref{fig:pairs_gaps}. It is easy to see that only in the case of RR, by selecting new interval $Y_i$ or $Z_i$, the number of selected intervals can increase, and by at most 1. Let $k_{RR}$ be the number of RR gap decisions, then the number of selected intervals by \ALG is $\le k + k_{RR}$. Two consecutive gap decisions cannot both be RR (otherwise, there would exist a pair $(I_i, J_i)$ such that both are rejected), implying $k_{RR} \leq \frac{k}{2}$. Since in each pair $(I_i, J_i)$ exactly one interval was selected, there must be an AA gap between every two RR gaps. Therefore, there are at least $2k_{RR} - 1$ gap decisions that are either RR or AA.
Let $n$ be the number of intervals. Then, 
	$n \ge 2k+(k-1)+(2k_{RR} -1) = 3k + 2k_{RR} - 2$.
Hence, the competitive ratio is 
	$\le 2(k+k_{RR})/n$,
satisfying
	$k_{RR} \le k/2$ and $3k + 2k_{RR} -2 \leq n$.
Optimizing this we get the upper bound $3/4$ as claimed. 
\end{proof}

\begin{theorem}
\label{thm:advice_lb}
    The advice complexity of  the Online Interval Selection on a simple chain is at least $n/4$.
\end{theorem}

\begin{proof} 
    The prefix of $k$ pairs leaves $k-1$ gaps, each to be connected with two choices. This results in $2^{k-1}$ different input sequences. For each input sequence, we extend the chain by adding intervals at the end, ensuring that all input sequences have exactly $4k+1$ intervals. 
    In each sequence, if we number the intervals $1, \ldots, 4k+1$, the optimal solution is to select the odd-numbered intervals. The parity assigned to the intervals in the prefix differs for each input sequence, hence $2^{k-1}$ distinct algorithms are required to make the correct decisions for the first $2k$ intervals. Hence, at least $k-1$ bits of advice are needed.
\end{proof}

\bibliography{mybib}{}
\bibliographystyle{alpha}

\end{document}